\documentclass[12pt]{amsart}

\usepackage{amssymb,amsmath}
\usepackage{amsthm}
\usepackage{bbm} 
\usepackage{amsfonts,mathrsfs}
\usepackage{url}
\usepackage{enumerate}
\usepackage{txfonts}

\def\ot{\otimes}
\def\D{\textsf{D}}\def\S{\textsf{S}}\def\T{\textsf{T}}

\newcommand{\out}[2]{| #1\rangle\langle #2 |}

\newcommand{\trans}{{\scriptstyle\mathsf{T}}}

\newcommand{\innerm}[3]{\langle #1 | #2 | #3 \rangle}

\newcommand{\pa}[1]{(#1)}
\newcommand{\Pa}[1]{\left(#1\right)}

\newcommand{\set}[1]{\{#1\}}

\newcommand{\ket}[1]{|#1\rangle}

\def\Jamiolkowski{J}
\newcommand{\jam}[1]{\Jamiolkowski\pa{#1}}

\DeclareMathOperator{\vectorize}{vec}
\newcommand{\col}[1]{\vectorize\pa{#1}}
\newcommand{\row}[1]{\vectorize\pa{#1}^{\dagger}}

\DeclareMathOperator{\trace}{Tr}
\newcommand{\ptr}[2]{\trace_{#1}\pa{#2}}
\newcommand{\Ptr}[2]{\trace_{#1}\Pa{#2}}
\newcommand{\tr}[1]{\ptr{}{#1}}
\newcommand{\Tr}[1]{\Ptr{}{#1}}

\newcommand{\tinyspace}{\mspace{1mu}}
\newcommand{\abs}[1]{|\tinyspace#1\tinyspace|}

\newcommand{\fontmapset}{\mathbf} 

\newcommand{\mset}[2]{\fontmapset{#1}\pa{#2}}

\newcommand{\lin}[1]{\mset{L}{#1}}

\newcommand{\identity}{\mathbbm{1}}

\newcommand{\Complex}{\mathbb{C}}

\def\cH{\mathcal{H}}
\def\cK{\mathcal{K}}

\newtheorem{thrm}{Theorem}[section]
\newtheorem{lem}[thrm]{Lemma}
\newtheorem{prop}[thrm]{Proposition}
\newtheorem{cor}[thrm]{Corollary}
\theoremstyle{definition}

\newtheorem{remark}[thrm]{Remark}

\numberwithin{equation}{section}

\author{Lin Zhang}
\address{Department of Mathematics\\
Zhejiang University\\Hangzhou\\
People's Republic of China}
\email{godyalin@163.com, linyz@zju.edu.cn}

\begin{document}

\title[Remark on Entropic characterization of quantum operations]{Remark on Entropic characterization \\of\\ quantum operations}\maketitle
\begin{abstract}
In the present paper, the reduction of some proofs in \cite{Roga1} is presented. An entropic inequality for quantum state and bi-stochastic CP super-operators is conjectured.
\end{abstract}


\section{Introduction}


For the notations, readers are refereed to \cite{Roga1}. So-called quantum operations are completely positive (CP) and trace-preserving (TP) linear super-operator $\Phi$, it is also called quantum channel, stochastic CP super-operator, by which the decoherence induced in an $d$-level quantum system $\cH_d$ are described by the map entropy $\S^{\mathrm{map}}(\Phi)$. When stochastic CP super-operator is unit-preserving, it is called bi-stochastic. In particular, depolarizing channels own an important role, for instance, Roga and his colleagues in \cite{Roga1} obtained a result which states that \emph{among all quantum channels with a given minimal Renyi output entropy of order two the depolarizing channel has the smallest map entropy Renyi entropy}. The present notes aims to give another proof, based on a theorem in \cite{Pedersen}. For convenience, it is listed as follows:


\begin{lem}(\cite{Pedersen})\label{average}
For any linear operator $M$ on an $n$-dimensional complex Hilbert space, the uniform average of $|\innerm{\psi}{M}{\psi}|^2$
over state vectors $\ket{\psi}$ on the unit sphere $S^{2n-1}$ in $\Complex^n$ is given by\\ \indent
$\int_{S^{2n-1}}|\innerm{\psi}{M}{\psi}|^2 d\mu(\psi)=\frac{1}{n(n+1)}[\tr{MM^\dagger}+|\tr{M}|^2]$,\\
where $d\mu$ is the normalized measure on the sphere.
\end{lem}

\section{Main Results}

\begin{prop}\label{prop:1}
If $\S_2^{\mathrm{min}}(\Lambda_n)=\S_2^{\mathrm{min}}(\Phi)$, then $\S_2^{\mathrm{map}}(\Lambda_n)\leqslant\S_2^{\mathrm{map}}(\Phi)$.
\end{prop}


\begin{proof}
Now let\\ \indent
$\S_2^{\mathrm{min}}(\Phi)=\S_2^{\mathrm{min}}(\Lambda_n)=-\log(1-\epsilon)\qquad\forall\epsilon\in(0,1)$.\\
Then\\ \indent
$\S_2^{\mathrm{map}}(\Lambda_n)=-\log[1-(1+\frac1n)\epsilon]$.\\
It suffices to prove $\S_2^{\mathrm{map}}(\Phi)\geqslant-\log[1-(1+\frac1n)\epsilon]$ whenever $\S_2^{\mathrm{min}}(\Phi)=-\log(1-\epsilon)$. This is equivalent to the following statement:\\ \indent
$\tr{[\Phi(\out{\varphi}{\varphi})]^2}\leqslant1-\epsilon \Longrightarrow \tr{[\rho(\Phi)]^2}\leqslant1-(1+\frac1n)\epsilon$,\\
where $\rho(\Phi)=\frac1n \jam{\Phi}$ is the Jamio{\l}kowski state and $\jam{\Phi}=\Phi\ot\identity(\out{\identity}{\identity})$ is the dynamical matrix of quantum channel $\Phi$.
Now by the Kraus decomposition of $\Phi$:\\ \indent
$\Phi(\sigma)=\sum_{i}K_i \sigma K_i^\dagger,\quad\sum_i K_i^\dagger K_i=\identity$.\\
Thus\\ \indent
$\tr{[\Phi(\out{\varphi}{\varphi})]^2}= \sum_{i,j}|\innerm{\varphi}{K_i^\dagger K_j}{\varphi}|^2\leqslant1-\epsilon,\quad \forall \ket{\varphi}$\\
and\\ \indent
$\tr{[\rho_{\Phi}]^2}=\frac{1}{n^2}\sum_{i,j}|\tr{K_i^\dagger K_j}|^2$.\\
It follows from the Lemma \ref{average} that
\begin{eqnarray}
\nonumber&&\int_{S^{2n-1}}\sum_{i,j}|\innerm{\varphi}{K_i^\dagger K_j}{\varphi}|^2 d\mu(\varphi)\leqslant(1-\epsilon)\int_{S^{2n-1}}d\mu(\varphi)=1-\epsilon\\
\nonumber&\Longleftrightarrow&
\sum_{i,j}\int_{S^{2n-1}}|\innerm{\varphi}{K_i^\dagger K_j}{\varphi}|^2 d\mu(\varphi)
\\
\nonumber&=&\frac{1}{n(n+1)}\sum_{i,j}\left[\tr{[K_i^\dagger K_j][K_i^\dagger K_j]^\dagger}+|\tr{K_i^\dagger K_j}|^2\right]
\\&=&\frac{1}{n(n+1)}\sum_{i,j}\left[\tr{K_i K_i^\dagger K_j K_j^\dagger}+|\tr{K_i^\dagger K_j}|^2\right]\leqslant1-\epsilon.\label{Eqn1}
\end{eqnarray}
Therefore it follows from Schwarcz's inequality that
\begin{eqnarray}
\nonumber n^2&=&[\tr{\identity}]^2=[\Tr{\sum_{i} K_i^\dagger K_i}]^2=[\Tr{\sum_{i}K_i K_i^\dagger}]^2\\
\nonumber&\leqslant&\tr{\identity}\Tr{\left[\sum_{i}K_i K_i^\dagger\right]^\dagger\left[\sum_{j}K_j K_j^\dagger\right]}\\
\nonumber&=&n\sum_{i,j}\tr{K_i K_i^\dagger K_j K_j^\dagger}\\
&\Longleftrightarrow& \sum_{i,j}\tr{K_i K_i^\dagger K_j K_j^\dagger}\geqslant n.\label{Eqn2}
\end{eqnarray}
Combining Eq.~(\ref{Eqn1}) with  Eq.~(\ref{Eqn2}) gives that\\ \indent
$\frac{1}{n^2}\sum_{i,j}|\tr{K_i^\dagger K_j}|^2\leqslant1-(1+\frac1n)\epsilon$.
\end{proof}


\begin{remark}
In \cite{Roga1}, Roga \emph{et al.} used the following fact that let $\mu$ be the Haar measure on the unitary group $\mathbf{U}(n)$ and let $A$ be a $n^2\times n^2$ matrix, then
$$\int_{\mathbf{U}(n)}U^{\ot 2}A(U^\dagger)^{\ot 2}d\mu(U)=\left(\frac{\tr{A}}{n^2-1}-
\frac{\tr{AS}}{n(n^2-1)}\right)^2\identity-\left(\frac{\tr{A}}{n(n^2-1)}-\frac{\tr{AS}}{n^2-1}\right)^2 S,$$
where $S$ stands for the swap operator: $S\ket{\mu\nu}=\ket{\nu\mu}$.
\end{remark}


\begin{cor}
If $\S_2^{\mathrm{map}}(\Lambda_n)=\S_2^{\mathrm{map}}(\Phi)$, then $\S_2^{\mathrm{min}}(\Lambda_n)\geqslant\S_2^{\mathrm{min}}(\Phi)$.
\end{cor}


\begin{proof}
Since the function $f(t)=-\log\left(\frac{1+n e^{-t}}{n+1}\right)(n\geqslant2)$ is a contiguous, increasing function, it follows that\\ \indent
$\S_2^{\mathrm{min}}(\Lambda_n)=-\log\left(\frac{1+n e^{-\S_2^{\mathrm{map}}(\Lambda_n)}}{n+1}\right)$;\\
i.e., $\S_2^{\mathrm{min}}(\Lambda_n)$ is a monotonic increasing function of $\S_2^{\mathrm{map}}(\Lambda_n)$. This fact together the above proposition \ref{prop:1} implies that the conclusion.
\end{proof}


It will be useful throughout the present paper to make use of a simple correspondence between the spaces $\lin{\cH,\cK}$ and $\cK\ot\cH$, for given Hilbert spaces $\cH$ and $\cK$. The mapping \\ \indent $\mbox{vec}:\lin{\cH,\cK}\longrightarrow\cK\ot\cH$\\
can be defined to be the linear mapping: \\ \indent$\col{A}=\sum_{m\mu}A_{m\mu}\ket{m}\ot\ket{\mu}=\sum_{m\mu}A_{m\mu}\ket{m\mu}$\\
for every operator $A\in\lin{\cH,\cK}$, where $\set{\ket{m}}$ and $\set{\ket{\mu}}$ are the standard basis for $\cH$ and $\cK$, respectively. When the vec mapping is generalized to multipartite spaces, caution should be given to the bipartite case (multipartite situation similarly). Specifically, for given complex Euclidean spaces $\cH_{A/B}$ and $\cK_{A/B}$,\\ \indent
$\mbox{vec}:\lin{\cH_{A}\ot\cH_{B},\cK_{A}\ot\cK_{B}}\longrightarrow
\cK_{A}\ot\cH_{A}\ot\cK_{B}\ot\cH_{B}$\\
is defined to be the linear mapping that represents a change of bases from the standard basis of $\lin{\cH_{A}\ot\cH_{B},\cK_{A}\ot\cK_{B}}$ to the standard basis of $\cK_{A}\ot\cH_{A}\ot\cK_{B}\ot\cH_{B}$. Concretely, \\ \indent
$\col{\out{m}{n}\ot\out{\mu}{\nu}}:=\ket{mn}\ot\ket{\mu\nu}$, \\ where $\set{\ket{n}}$ is the standard basis for $\cH_{A}$ and $\set{\ket{\nu}}$ is the standard basis for $\cH_{B}$, while $\set{\ket{m}}$ is the standard basis for $\cK_{A}$ and $\set{\ket{\mu}}$ is the standard basis for $\cK_{B}$. The mapping is determined for every operator $X\in\lin{\cH_{A}\ot\cH_{B},\cK_{A}\ot\cK_{B}}$ by linearity. Note that if $X=A\ot B$, where $A\in\lin{\cH_{A},\cK_{A}}$ and $B\in\lin{\cH_{B},\cK_{B}}$, then $\col{A\ot B}=\col{A}\ot\col{B}$.


\begin{prop}\label{product channel}
For any CP super-operators $\Phi$ and $\Psi$, with corresponding their Kraus representations: $\Phi=\sum_{i}Ad_{M_{i}}$ and $\Psi=\sum_{j}Ad_{N_{j}}$, respectively. Then:
\begin{enumerate}
\item $\Phi\ot\Psi=\sum_{i,j}Ad_{M_{i}\ot N_{j}}$, and  $\jam{\Phi\ot\Psi}=\jam{\Phi}\ot \jam{\Psi}$;
\item $\jam{\Phi\circ\Psi}=\Phi\ot\identity(\jam{\Psi})=\identity\ot\Psi^\trans(\jam{\Phi})
    =\Phi\ot\Psi^\trans(\col{\identity}\row{\identity})$.
\end{enumerate}
\end{prop}


\begin{prop}
Let $\Phi$ and $\Psi$ are CP stochastic super-operators. For any $p\geqslant0$, the Renyi map entropy satisfies the following addition identity:\\ \indent
$\S^{\mathrm{map}}_p(\Phi\ot\Psi)=\S^{\mathrm{map}}_p(\Phi)+\S^{\mathrm{map}}_p(\Psi)$.
\end{prop}


\begin{proof}
By the Lemma \ref{product channel}, $\rho(\Phi\ot\Psi)=\rho(\Phi)\ot\rho(\Psi)$, which implies that $\|\rho(\Phi)\ot\rho(\Psi)\|_p=\|\rho(\Phi)\|_p\|\rho(\Psi)\|_p$ since $\|\cdot\|_p$ is multiplicative.
Now
\begin{eqnarray*}
\S^{\mathrm{map}}_p(\Phi\ot\Psi)&=&\frac{p}{1-p}\log\|\rho(\Phi)\ot\rho(\Psi)\|_p\\
&=&\frac{p}{1-p}\log\|\rho(\Phi)\|_p+\frac{p}{1-p}\log\|\rho(\Psi)\|_p\\
&=&\S^{\mathrm{map}}_p(\Phi)+\S^{\mathrm{map}}_p(\Psi).
\end{eqnarray*}
\end{proof}


\begin{lem}\label{dualchannel}(\cite{Lin})
Let $\Phi$ be CP super-operator with their Kraus operators $\set{K_i}$. The corresponding transpose of $\Phi$ is $\Phi^{\trans}$ for which the Kraus operators $\sec{K_i^{\trans}}$; the dual of $\Phi$ is $\Phi^\dagger$ for which the Kraus operators $\set{K_i^\dagger}$. Then\\ \indent
$\jam{\Phi^\trans}=S\jam{\Phi}S$ and $\jam{\Phi^{\dagger}}=S\jam{\Phi}^{\trans}S$,\\
where $S$ is the swap operator.
\end{lem}


\begin{lem}(\cite{Roga2})
Let $\Phi,\Psi$ be CP bi-stochastic super-operators. Then the dynamical subadditivity inequality is valid:\\ \indent
$\max\set{\S^{\mathrm{map}}(\Phi),\S^{\mathrm{map}}(\Psi)}\leqslant
\S^{\mathrm{map}}(\Phi\circ\Psi)\leqslant\S^{\mathrm{map}}(\Phi)+\S^{\mathrm{map}}(\Psi)$.
\end{lem}


\begin{proof} In the present proof, a simple approach is given, which is different from the one in \cite{Roga2}. Specifically, since $\Phi$ is bi-stochastic, it is easily seen that \\ \indent
$\S(\Phi\circ\Psi)=\S(\Phi\ot\identity(\rho_{\Psi}))\geqslant\S(\Psi),
\S(\Phi\circ\Psi)=\S(\identity\ot\Psi^\trans(\rho_{\Phi}))\geqslant\S(\Phi)$.\\
Now from the Lindblad's entropic inequality; i.e.,\\ \indent
$\abs{\S(\hat{\sigma}(\rho,\Lambda))-\S(\rho)}\leqslant
\S(\Lambda(\rho))\leqslant\S(\hat{\sigma}(\rho,\Lambda))+\S(\rho)$,\\
it follows that\\ \indent
$\S(\Phi\ot\identity(\rho_{\Psi}))\leqslant\S(\rho_{\Psi})+
\S(\hat{\sigma}(\rho_{\Psi},\Phi\ot\identity))\Longleftrightarrow
\S(\Phi\circ\Psi)\leqslant\S(\Phi)+\S(\Psi)$, \\
where $\S(\hat{\sigma}(\rho_{\Psi},\Phi\ot\identity))=\S(\Phi)$ can be easily checked.
\end{proof}


\begin{prop}
For any bi-stochastic CP super-operators $\Phi,\Psi$ and any maximally entangled state $\frac{1}{\sqrt{n}}\col{U}$, where $U$ is a $n\times n$ complex unitary matrix, the following inequality for von Neumann entropy holds:\\ \indent
$\max\set{\S^{\mathrm{map}}(\Phi),\S^{\mathrm{map}}(\Psi)}\leqslant
\S((\Phi\ot\Psi)(\col{U}\row{U}))\leqslant\S^{\mathrm{map}}(\Phi)+\S^{\mathrm{map}}(\Psi)$
\end{prop}


\begin{proof}
It suffices to show that when $U=\identity$. Thus let $U=\identity$. Clearly, it follows from the Lemma \ref{dualchannel} that $\S^{\mathrm{map}}(\Psi)=\S^{\mathrm{map}}(\Psi^\trans)$ when $\Psi$ are bi-stochastic (Therefore so does $\Psi^\trans$). Since\\ \indent
$(\Phi\ot\Psi)(\col{\identity}\row{\identity})=
(\Phi\circ\Psi^{\trans}\ot\identity)(\col{\identity}\row{\identity})=\jam{\Phi\circ\Psi^\trans}$,\\
which implies that \\ \indent
$\S((\Phi\ot\Psi)(\col{\identity}\row{\identity}))=\S^{\mathrm{map}}(\Phi\circ\Psi^\trans)$.\\
Now under the condition that $\Phi,\Psi$ are bi-stochastic, it is obtained from \cite{Roga2} that\\ \indent
$|\S^{\mathrm{map}}(\Phi)-\S^{\mathrm{map}}(\Psi)|\leqslant\max\set{\S^{\mathrm{map}}(\Phi),\S^{\mathrm{map}}(\Psi)}
\leqslant\S^{\mathrm{map}}(\Phi\circ\Psi^\trans)
\leqslant\S^{\mathrm{map}}(\Phi)+\S^{\mathrm{map}}(\Psi)$.
\end{proof}


\begin{remark}
If the Lindblad's inequality employed, then
\begin{eqnarray*}
|\S^{\mathrm{map}}(\Phi)-\S^{\mathrm{map}}(\Psi)|&=&|\S(\rho(\Phi))-\S(\widehat{\sigma}[\identity\ot\Psi,\rho(\Phi)])|
\leqslant\S((\identity\ot\Psi)(\rho(\Phi)))\\
&\leqslant&\S(\rho(\Phi))+\S(\widehat{\sigma}[\identity\ot\Psi,\rho(\Phi)])
=\S^{\mathrm{map}}(\Phi)+\S^{\mathrm{map}}(\Psi).
\end{eqnarray*}
Thus the result can be generalized to the CP stochastic maps \cite{Roga1}.
\end{remark}


\begin{remark}
There is a \emph{conjecture} which can be stated as follows: if $\Phi,\Psi\in\T(\cH)$ are bi-stochastic CP super-operators, then \\ \indent
$\S(\rho)+\S(\Phi\circ\Psi(\rho))\leqslant\S(\Phi(\rho))+\S(\Psi(\rho))$ \\
for any $\rho\in\D(\cH)$, where $\D(\cH)$ stands for all density matrix acting on a $N$-dimensional Hilbert space $\cH$ and $\T(\cH)$ all linear super-operators from $\lin{\cH}$ to $\lin{\cH}$.
And what is a characterization of the saturation for the above inequality. If this conjecture is correct, then it can be employed to give a simple proof to the strong dynamical subadditivity for bi-stochastic CP super-operators.
\end{remark}



\end{document}